\documentclass[10pt]{article}
\usepackage[english]{babel}
\usepackage{amsmath,amsfonts,amssymb,amstext,amsmath,amsthm,amscd}

\usepackage[a4paper,top=2.5cm,bottom=3cm,left=3cm,right=4cm]{geometry}

\usepackage{booktabs}
\usepackage{xcolor}
\usepackage{mathtools}
\usepackage{mathrsfs}
\usepackage{dsfont}

\usepackage{graphicx}
\usepackage{subfigure}
\usepackage{wrapfig}
\usepackage{indentfirst}
\usepackage{latexsym}
\usepackage{sidecap}
\usepackage{booktabs}

\usepackage{setspace}
\usepackage{calrsfs}
\usepackage{hyperref}
\hypersetup{
    colorlinks=true,
    linkcolor=red,
    filecolor=magenta,      
    urlcolor=cyan,
}
\usepackage{mathrsfs}
\usepackage{textcomp}
\usepackage{braket}
\usepackage{colortbl}
\usepackage{bbm}
\usepackage{comment}
\usepackage[colorinlistoftodos]{todonotes}

\theoremstyle{plain}
\newtheorem{Theorem}{Theorem}[section]
\theoremstyle{definition}
\newtheorem{Definition}{Definition}[Theorem]
\theoremstyle{remark}
\newtheorem{Remark}[Theorem]{Remark}
\theoremstyle{remark}

\theoremstyle{plain}
\newtheorem{Lemma}[Theorem]{Lemma}
\theoremstyle{plain}
\newtheorem{Corollary}[Theorem]{Corollary}
\theoremstyle{plain}
\newtheorem{Proposition}[Theorem]{Proposition}
\theoremstyle{remark}

\theoremstyle{remark}

\theoremstyle{remark}
\newtheorem{Hypothesis}[Theorem]{Hypothesis}
\theoremstyle{remark}

\newcommand\A{\mathcal{A}}

\newcommand\C{\mathcal{C}}

\newcommand\RR{\mathbb{R}}

\newcommand\ind{\mathbbm{1}}

\DeclareMathAlphabet{\pazocal}{OMS}{zplm}{m}{n}

\newcommand{\ve}{\varepsilon}

\usepackage{tikz}
\usetikzlibrary{arrows}

\newcommand{\R}{\mathbb{R}}

\newcommand{\Ind}[1]{\mathbbm{1}_{\left \{#1\right \}}}

\newcommand{\ds}{\displaystyle}

\newcommand\I{\mathcal{I}}

\usepackage{tikz}
\usetikzlibrary{arrows}

\makeatletter
\pgfdeclareshape{datastore}{
  \inheritsavedanchors[from=rectangle]
  \inheritanchorborder[from=rectangle]
  \inheritanchor[from=rectangle]{center}
  \inheritanchor[from=rectangle]{base}
  \inheritanchor[from=rectangle]{north}
  \inheritanchor[from=rectangle]{north east}
  \inheritanchor[from=rectangle]{east}
  \inheritanchor[from=rectangle]{south east}
  \inheritanchor[from=rectangle]{south}
  \inheritanchor[from=rectangle]{south west}
  \inheritanchor[from=rectangle]{west}
  \inheritanchor[from=rectangle]{north west}
  \backgroundpath{
    \southwest \pgf@xa=\pgf@x \pgf@ya=\pgf@y
    \northeast \pgf@xb=\pgf@x \pgf@yb=\pgf@y
    \pgfpathmoveto{\pgfpoint{\pgf@xa}{\pgf@ya}}
    \pgfpathlineto{\pgfpoint{\pgf@xb}{\pgf@ya}}
    \pgfpathmoveto{\pgfpoint{\pgf@xa}{\pgf@yb}}
    \pgfpathlineto{\pgfpoint{\pgf@xb}{\pgf@yb}}
 }
}
\makeatother

\usepackage{fancyhdr}
\usepackage{calc} 
\begin{document}

\title{A bank salvage model by impulse stochastic controls}
\author{Francesco Cordoni$^{a}$ \and Luca Di Persio$^{a}$ \and Yilun Jiang$^{b}$}
\date{}
\maketitle

\renewcommand{\thefootnote}{\fnsymbol{footnote}}
\footnotetext{{\scriptsize $^{a}$ Department of Computer Science, University of Verona, Strada le Grazie, 15, Verona, 37134, Italy}}
\footnotetext{{\scriptsize E-mail addresses: francescogiuseppe.cordoni@univr.it
(Francesco Cordoni), luca.dipersio@univr.it (Luca Di Persio)}}

\footnotetext{{\scriptsize $^{b}$ Department of Mathematics, Penn State University, University Park, PA 16802, USA}}
\footnotetext{{\scriptsize E-mail addresses: sjtujyl@gmail.com (Yilun Jiang)}}

\begin{abstract}
The present paper is devoted to the study of a {\it  bank salvage model} with finite time horizon and subjected to stochastic impulse controls. In our model, the bank's default time is a completely inaccessible random quantity generating its own filtration, then reflecting  the unpredictability of the event itself. In this framework the main goal is to minimize the total cost of the central controller who can inject capitals to save the bank from default. We address the latter task showing that the corresponding \textit{quasi--variational inequality} (QVI) admits a unique viscosity solution, Lipschitz continuous in space and H\"older continuous in time. Furthermore, under mild assumptions on the dynamics the smooth-fit $W^{(1,2),p}_{loc}$ property is achieved for any $1<p<+\infty$.
\end{abstract}

\textbf{AMS Classification subjects:} 49N25, 49N60, 93E20, 91G80 \medskip

\textbf{Keywords or phrases: } Bank salvage model, stochastic impulse control, viscosity solution, inaccessible bankruptcy time, smooth-fit property.

\maketitle

\section{Introduction} 

Mainly motivated by the recent financial credit crisis, starting from 2008-2009 credit crunch, the financial and mathematical community started investigating and generalize existing models, since previous events have shown that financial models used prior to the crisis where inadequate to describe and capture main features of financial markets. Therefore the mathematical and financial communities have focus on developing general and robust models that are able to properly describe financial markets and their main peculiarities. 
%Particular attention has been paid to consider a financial supervisor that is willing to control financial entities, possibly lending money in order to avoid failure, see, e.g. \cite{Cap,CDPP,Eis,Lip,Rog}.

From a purely mathematical perspective the above mentioned attention led, among many other research topics, for instance into the study of general stochastic optimal control problems, where instead of classical type of controls, some more realistic controls have been considered. Among the most studied type, impulse type controls have to be mentioned, and regained attention in last decades also due to the many application in finance and economics. In this setting, the controller can intervene on the system at some random time with a discrete type control, where in this case the control solution is represented by the couple $u = (\tau_n,K_n)_n$, where $\tau_n$ represents the decision time at which the ocntroller intervene and $K_n$ instead denotes the action taken by the controller. Above type of control implies that at the intervention time $\tau_n$ the system jumps from the state $X(\tau_n^-)$ to the new state $X(\tau_n) = \Gamma(X(\tau_n),K_n)$, for a suitable function $\Gamma$. Therefore, as standard in optimal control theory, using the dynamic programming principle, it can be shown that \textit{stochastic impulse control} problems can be associated to a quasi--variational Hamilton--Jacobi--Bellman equation (HJB) of the form
\begin{equation}\label{EQN:HJB}
\min \left [- \frac{\partial}{\partial t} V - \mathcal{L}V - f, V- \mathcal{H}V\right ]\,,
\end{equation}
where above $f$ is the running cost, $\mathcal{L}$ is the infinitesimal generator for the process $X$ and $V$ is the value function solution to the above HJB equation. Further, $\mathcal{H}$ is the \textit{nonlocal impulse operator} that characterize HJB equation for impulse type of control. The particular form for the HJB implies that two regions can be retrieved, the \textit{continuation region} where $V > \mathcal{H}V$ and therefore no impulse control is used, and the \textit{impulse region} where on the contrary $V = \mathcal{H}V$ and the controller intervenes. Solution to equation \eqref{EQN:HJB} can be formally defined so that the value function is in fact a viscosity solution, in a sense to be properly defined later on, to equation \eqref{EQN:HJB}. It is clear that, following the above characterization of the domains for the HJB equation, particular attention must be given to the \textit{intervention boundary}. In fact, particular attention is usually given in this field to proving that the boundary is regular enough; this regularity is referred to in literature as \textit{smooth--fit principle}. Several results exist in the smooth--fit principle where the terminal horizon for the control problem, whereas instead finite horizon problem, and in particular the terminal condition of the problem, makes less straightforward the derivation of the smooth--fit principle.

At last, we stress that impulse type stochastic control is strictly connected to optimal stopping problems and optimal switching. The literature on the topic is wide, we refer the interested reader to \cite{GuoChen,Tang}, or also to \cite{Bay,BCS,Che,Che2,Ega,GuoWu,Oks,Pha3,Vat} for other related results.

A second crucial financial aspect that emerged to be fundamental in a general financial formulation after recent crisis there is possible failures of financial entities. In fact, one of the major lack of classical financial models is that no risk of failure is considered into the general setting. Recent financial event has shown that no financial operator can be considered immune from bankruptcy. Therefore it has emerged in last decade an extensive literature that focus on credit risk modeling, assessing as main object the risk that financial entities has to face borrowing or lending money to other players that might fail, see, e.g. \cite{BP,Cre1,Cre2}.

Along aforementioned lines, two main approaches have been developed in literature: {\it structural approach} and {\it intensity--based approach}, see, e.g. \cite{Bie}. Mathematically speaking, the first scenario consits in considering some default event that can be triggered by the underlyng process. Typical example are default triggered by some stopping time defined as a hitting time. Such an approach has been for instance considered in \cite{CDP,CDPP,Lip}. The latter instead considers a default event which is completely inaccessible for the probabilistic reference filtration, so that in order to solve the problem the typical approach is to rely on {\it filtration enlargment} techniques, see, e.g., \cite{BR,Pha}.

The present paper is devoted to study a stochastic optimal control problem of impulse type, where a financial supervisor controls a system, such as financial operators or also some banks. The final goal of the controller is to prevent failures, injecting capital intro the system according to a given criterion to be maximized. The controller has no perfect information regarding the failure of the bank, so that mathematically speaking the failure cannot be foreseen by the controller. The supervisor, which can be though for instance as a central bank, can intervene with some impulse type controls over a finite horizon, so that the optimal solution is represented by both am intervention time and the quantity to inject into the system.

Our approach will be based on a \textit{intensity--based approach}, so that we will assume the default event to be totally inaccessible from the reference filtration, assuming only a typical density assumption. This assumptions will allows us to rewrite the system as deterministic finite horizon impulse problem, using the density distribution of the default event, via \textit{enlargement of filtrations} techniques. We stress that, due to the terminal condition to be imposed, typically a finite horizon stochastic impulse control problem is more difficult to solve than infinite horizon impulse control problems. In fact, it exists an exhaustive literature on stochastic impulse control on infinite time horizon, see, e.g. \cite{Bay,BCS,Che,Ega,GuoWu,Oks,Pha3,Vat}, whereas very few results exist for the finite dimensional case, see, e.g. \cite{Che2,GuoChen,Tang}.

A more financially oriented motivation of the control problem considered in the present work, has often arise in the last decade, mostly  as a consequence of the 2007--2008 credit crunch. This has been for instance the case of {\it Lehman Brothers} failure, which has shown the cascade effect triggered by the default of a sufficiently large and interconnected financial institution, see, e.g., \cite{IvashinaScharfstein2010, KahleStulz2013} and references therein.
We stress that, particular attention has to be given not only towards the magnitude of the stressed bank's financial assets,  but also to its interconnection grade. Indeed, while the exposure with few financial institutions, provided its magnitude is reasonable, can be managed by {\it ad hoc} politics established on a {\it one to one} relationship basis, the situation could be simply ungovernable in case of a high number of connections, hidden links and over-structured contracts.

Since above mentioned financial crisis, it has became typical, within the financial oriented stochastic optimal control  theory, to model a given problem up to a random terminal time instead of considering a fixed, even infinite, horizon. 
From a modelling point of view, the aforementioned scenario has lead to consider the stochastic optimal control approach  to model such situations by considering random terminal times, instead of considering a fixed, or infinite, horizon. 
Analogously, data analysts as well as mathematicians, have started to consider problems of bank bailouts, where bank's default and the consequent contagion spreading inside the network, may induce serious consequences for decades, see, e.g., \cite{EMNS2012}.

From a government perspective, such type of likely high financial fall out, have pushed several central banks to establish specific economic actions to help those sectors of the banking sector of (at least) national interest, under concrete  failure risks. As an example, the latter has been the case of the pro bail-in procedures followed in agreement with the  Directive 2014/59/UE (approved last $1^{st}$ of January, 2016 by the European Union Parliament), and then applied, e.g., in Italy, Ukraine, etc., see, e.g. \cite{KAO2019,SakuramotoUrbani2018}. 
It is relevant to underline that such actions rely also on the following grades of freedom: the possibility, as an alternative to internal rescue, to relocate goods as well as  legal links to a third party, often called bridge-bank, or  to a bad bank which will  collect only a part of assets aiming at maximizing its long-term value; the hierarchical order of those who are called to bear the bail-in, which means that the government can decide to put small creditors on the safe side; and the principle that no shareholder, or creditor, has to bear greater losses than would be expected if there was an administrative liquidation, namely the no worse off creditor idea.

Similar situations have been recently taken into consideration by a series the Central European Bank procedures, with particular reference to the well known \textit{quantitative easing}, as well as in agreement to the creation of \textit{injected currency}, see, e.g., \cite{ECB3,ECB1,ECB4,ECB2}. We would like to underline that quantitative easing type procedures have been experienced also outside the European Union, as in the case of  the actions undertaken by the Japanese Central Bank, whose  intervention has lasted over years, see, e.g. \cite{Jap1,Jap2,Jap3}, or how has been done by the US Federal Reserve not only starting from  2008, but also during the Great Depression of the 1930s, see, e.g., \cite{USA1,USA2,USA3,USA4,USA5}.

The main contribution of the present paper is to develop a concrete financial setting that models the evolution of a financial entity, controlled by an external supervisor who is willing to lend money in order to maximize a given utility function; see also \cite{Cap,CDPP,Eis,Lip,Rog} for setting in which a financial supervisor aims at controlling a system of banks of general financial entities. In compete generality, we will assume that the financial entity may fail at some random time that is inaccessible to the reference filtration, which represents the controller knowledge. Also, we consider a controller that can act on a system with an impulse--type control, so that the optimal solution consists in both a random time at which injecting money into the system and the precise amount of money to inject. We characterize the value function of the above problem, showing that it must solve in a given viscosity sense a certain quasi--variational inequality (QVI). At last we will prove that above QVI admits a unique solution in a viscosity sense and also we provide a regularity results for the intervention boundary, known in literature as \textit{smooth fit principle}.

The paper is so organized, Section \ref{SEC:GS} introduces the general financial and mathematical setting; then Section \ref{SEC:Reg} prove some regularity results for the value function and Section \ref{SEC:Vis} address the problem of existence and uniqueness of a solution. At last Section \ref{SEC:SF} is devoted to the smooth fit principle.

\section{The general setting}\label{SEC:GS}

We will in what follows consider a complete filtered probability space $\left (\Omega,\mathcal{F},\left (\mathcal{F}_t\right )_{t \in [0,T]},\mathbb{P}\right )$, $\left (\mathcal{F}_t\right )_{t \in [0,T]}$ being a filtration satisfying the usual assumptions, namely right--continuity and saturation by $\mathbb{P}$--null sets. Let $T<\infty$ be a fixed terminal time, and let $x$, resp. $y$, denotes the total value of the investments of a given bank, resp.  the total amount of deposit of the same bank. We assume that $x$ and $y$  evolve according to the following system of SDEs
\begin{equation}\label{EQN:XY}
\begin{cases}
dx(t) = c_1 \dot{y}(t)dt +\tilde{\mu}(t)x(t)dt+\tilde{\sigma}(t)x(t) dW(t)\\
\dot{y}(t) = \lambda\left (\frac{x(t)}{y(t)}\right )y(t)
\end{cases}
\;,
\end{equation}
where $W(t)$ is assumed to be a standard Brownian motion adapted to the aforementioned filtration. In particular, the first term in equation \eqref{EQN:XY} accounts for the increase in $X$ due to the fact that new deposits are made, where $c_1 \in [0,1]$ denotes the fractions of deposits which are actually invested in more or less risky financial operations. We stress that by a rescaling argument, with no less of generality $c_1=1$ it can be assumed. Moreover we define the {\it value over liability ratio} $X(t):= \frac{x(t)}{y(t)}$. Then, according to eq. \eqref{EQN:XY} and exploiting the It\^{o}-D\"{o}blin formula, we have
\begin{equation}\label{EQN:EvEq}
\begin{cases}
d X(t) = \left ((c_1 -X(t))\lambda(X(t)) + \tilde{\mu}(t)X(t)\right )dt + \tilde{\sigma}(t) X(t)dW(t) \\
x(0)= x_0
\end{cases}\; .
\end{equation}

We assume the process $X$ to be stopped at {\it completely inaccessible random time} $\tau$, not adapted to the reference filtration $\left (\mathcal{F}_t\right )_{t \in [0,T]}$. From a financial point of view, assuming that $X$ represents the financial value of an agent, above assumption reflects the fact that a bank's failure cannot be predicted. In particular, let us introduce the filtration $\left (\mathcal{H}_t\right )_{t \in [0,T]}$ generated by the stopping time $\tau$, namely $\mathcal{H}_t := \ind_{\{\tau \leq t\}}$. Then we define the augmented filtration $\left (\mathcal{G}_t\right )_{t \in [0,T]}$, where $\mathcal{G}_t := \mathcal{F}_t \vee \mathcal{H}_t$. 

Within this  setting it is interesting to consider an {\it external controller}, e.g., a {\it central bank}, or an equivalent  financial agent acting as a governance institution,  with suitable surveillance rights. Such controller can inject capital in the bank, at random times $\tau_n$. Then, at that time $\tau_n$, the state process $X(t)$ {\it jumps}, in particular we have 
\[
X(\tau_n^-) \not= X(\tau_n) =X(\tau_n^-)+K_n\; ,
\]
therefore $X(t)$ evolves according to
\begin{equation}\label{EQN:EvEqT}
\begin{cases}
d X(t) = \left ((c_1 -X(t))\lambda(X(t)) + \tilde{\mu}(t)X(t)\right )dt + \tilde{\sigma}(t) X(t)dW(t) + \sum_{n:\tau_n \leq t} K_n \\
X(0)= x_0
\end{cases}
\;.
\end{equation}

The solution to the aforementioned system is represented by a couple $u = \left (\tau_n,K_n\right )_{n \geq 1}$, where $(\tau_n)_{n \geq 1}$ is a non--decreasing sequence of stopping times representing  the {\it intervention times}, while  $(K_n)_{n\geq 0}$ is a sequence of $\left (\mathcal{G}_t\right )$--adapted random variables taking values in $\mathcal{A}\subset [0,\infty)$. In particular the sequence $(K_n)_{n\geq 0}$ indicates the  {\it financial actions} taken at time $\tau_n$. The following is the definition of \textit{admissible impulse strategy} $u$.

\begin{Definition}[Admissible impulse strategy]
The admissible control set $\mathcal{U}$ consists of all the impulse controls $u = \left (\tau_n,K_n\right )_{n \geq 0}$ such that
\begin{equation}\label{ADM}
	\begin{array}{l}
\{\tau_i\}_{i\geq 1} ~\hbox{are}~\mathcal{G}_t~\hbox{adapted stopping times and increasing, i.e }	\tau_1<\tau_2<\cdots<\tau_i<\cdots,\\[3mm]
	 K_i\in {\cal A}~\hbox{and}~K_i\in\mathcal{G}_{\tau_i},~\forall~i\geq 1.\\[3mm]
\end{array}
\end{equation}
\end{Definition}

\begin{Remark}
Equivalently, we will use a different notation, $\xi_t(\cdot)$ to express the same space, i.e. for all $0\leq t\leq s\leq T$,
$$\xi_t(s) = \sum_{t\leq\tau_i<s}K_i$$ 
where $\tau_i,K_i$ satisfies \eqref{ADM} and the corresponding admissible control set $\mathcal{U}[t,T]$ consists of all $\xi_t(\cdot)$.
\end{Remark}

In what follows we will denote for short
\[
\begin{array}{l}
\mu(t,X_t) :=  \left ((c_1 -X(t))\lambda(X(t)) + \tilde{\mu}(t)X(t)\right )\,,\\\\[3mm]
\sigma(t,X_t) :=  \tilde{\sigma}(t) X(t)\,,
\end{array}
\]

so that, for any admissible control $u\in\mathcal{U}[t,T]$, define $X^u_{t,x}(s)$ as
$$X^u_{t,x}(s) = x+\int_t^s\mu(r,X^u_{t,x}(r))\,dr+\int_t^s\sigma(r,X^u_{t,x}(r))d W(r)+\xi_t(s)\quad s\geq t,$$
which is the unique strong solution of dynamics \eqref{EQN:EvEqT} with initial condition $X^u_{t,x}(t) = x$. We aim at solving the following stochastic control problem whose value function is defined as
\begin{equation}
V(t,x) \doteq \sup_{u\in\mathcal{U}[t,T]}J^u(t,x)\;.
	\label{EQN:val}
\end{equation}
where $J^u(t,x)$ is the expected cost of the form

\begin{align}
J^u(t,x) \doteq \mathbb E  \bigg [\int_t^{\tau \wedge T} & f(X^u_{t,x}(s)) ds + g_1(X^u_{t,x}(T))\Ind{\tau \geq T}+\label{EQN:GenCP}\\
& -g_2(X^u_{t,x}(\tau))\Ind{\tau < T} - \sum_{t \leq \tau_n \leq \tau \wedge T}\left ( K_n+\kappa\right ) \bigg] \; ,\nonumber\\
\end{align}

where $f$, resp. $g$, represents the \textit{running cost}, resp. the \textit{terminal cost}, while  $K+\kappa$, $\kappa >0$, is a suitable constant defining the cost required by the capital injection. Above, we have denoted by $\tau$ the bank {\it default time}, with respect to the process $X(t)$. We assume, as specified above, that $\tau$ is a {\it completely inaccessible random time}, and it is not adapted to the reference filtration $\left (\mathcal{F}_t\right )_{t \in [0,T]}$. Also, recall that $\left (\mathcal{H}_t\right )_{t \in [0,T]}$ is the filtration generated by the stopping time $\tau$, namely $\mathcal{H}_t := \ind_{\{\tau \leq t\}}$, whilst $\left (\mathcal{G}_t\right )_{t \in [0,T]}$ is the filtration, namely $\mathcal{G}_t := \mathcal{F}_t \vee \mathcal{H}_t$. 

Following the standard literature, see, e.g., \cite{Karatzas} both the running and terminal costs are usually given in terms of  suitable utility functions representing the {\it utility gains} from the bank's value. A  typical example is $f(x) = \frac{x^p}{p}$, $p\in (0,1)$. 
As regards the cost $K+\kappa$, it reflects the fact that injecting an amount  $K$ of capital to increase the bank's liquidity level, implies a non negligible cost,  otherwise such a financial help would be always profitable.

Throughout the work we will make the following assumptions:
\begin{Hypothesis}\label{ASS:Lpz}
\begin{description}
\item[(i)] the function $\lambda:\RR \to \RR$ is Lipschitz continuous, namely there exists a constant $L_\lambda >0$ such that
\[
|\lambda(x)-\lambda(y)| \leq L_\lambda |x-y|\, ,\quad \forall \, x,\,y\in \RR\,;
\]

\item[(ii)] the functions $f$, $g_1$, $g_2$ are Lipschitz continuous, namely there exist constants $L_f$, $l_{g_1}$ and $L_{g_2} >0$ such that

\begin{align}
&|f(x)-f(y)| \leq L_f |x-y|\, ,\quad \forall \, x,\,y\in \RR\,;\nonumber \\
&|g_1(x)-g_1(y)| \leq L_{g_1} |x-y|\, ,\quad \forall \, x,\,y\in \RR\,;\nonumber \\
&|g_2(x)-g_2(y)| \leq L_{g_2} |x-y|\, ,\quad \forall \, x,\,y\in \RR\,;\nonumber \\
\end{align}

We also assume that there exist constants $C_f$, $C_{g_1}$ and $C_{g_2} >0$ such that

\begin{align}
&f(x) < C_f\,,\quad g_1(x) < C_{g_1}\,,\quad g_2(x)< C_{g_2}\,\nonumber;
\end{align}

\item[(iii)] the functions $\mu(t),\sigma(t) \in C([0,T])$. 
% is Lipschitz continuous, namely there exists $L_\mu>0,L_\sigma>0$ such that
%$$
%\begin{array}{rl}
%&|\mu(t_1)-\mu(t_2)|\leq L_\mu|t_1-t_2|,\quad \forall t_1,t_2\in[0,T];\\
%&|\sigma(t_1)-\sigma(t_2)|\leq L_\sigma|t_1-t_2|,\quad \forall t_1,t_2\in[0,T];\\
%\end{array}
%$$
\item[(iv)] No terminal impulse, i.e.
$$
g_1(x)\geq \sup_{K>0} g_1(x+K)-K-\kappa.
$$
\end{description}
\end{Hypothesis}

%\begin{Remark}

The boundedness properties for the running and terminal cost can be interpreted in the following sense: since we are seeking the optimal capital injection strategy for the government over a finite time horizon, we may think that there is a {\it healthy} level $U>0$ such that when the bank's capital is growing to infinity, then the utility remains flat, so that the government will have no interest in injecting more capital. As to make an example, we can take
\[
U(x) = U-e^{-x}\,.
\]
%\end{Remark}

%\begin{Remark}
%We would like to underline that the above setting  can be reformulated with respect to the  {\it infinite horizon} hypothesis. In such a case, the cost
%functional in equation \eqref{EQN:GenCP}, reads as follow
%\begin{equation}\label{EQN:GenCPIH}
%J(x) = \mathbb E \left [\int_0^{\tau} e^{-r s} f(x(s)) ds - e^{-r \tau} g(x(\tau)) - \sum_{\tau_n \leq \tau} e^{-r \tau_n} \left (K_n + \kappa\right )\right ] \; .\\
%\end{equation}
%\end{Remark}

\begin{Remark}
A further generalization of the above optimal control problem, consists in considering  a controller having  two different ways to influence the evolution of the state process $x$, namely
\begin{itemize}
\item[(1)]an impulse type control $(\tau_n,K_n)_n$, hence as in equation \eqref{EQN:EvEqT} by injecting capital at random times $\tau_n$;
\item[(2)] a continuous type control $\alpha(t)$,  by choosing at any time $t$ the rate at which $x$ is growing. 
\end{itemize}
In particular an action of type $2$ implies that eq. \eqref{EQN:EvEqT} can be reformulated as follows
\[
\begin{cases}
d X(t) = \left ((c_1 -X(t))\lambda(X(t)) + (\mu(t)- \alpha(t)) X(t)\right )dt + \sigma(t) X(t)dW(t) + \sum_{n:\tau_n \leq t} K_n\\
X(0)= x_0
\end{cases}
\; ,
\]
where $\alpha$ represents the continuous control variable $\alpha(t) \in [0,\bar{r}]$, for a suitable constant $\bar{r}$, where $\alpha=0$ stands for higher  returns and  $\alpha=\bar{r}$ denotes lower returns. This reflects the financial assumption that the controller, e.g. a {\it central bank}, can change the interest rate according to macroeconomic variables, as the country inflation level, the forecast of supranational interest rates, the  the markets' belief about the health of the financial sector under the central bank control, etc. In fact, choosing $\alpha=0$ the bank value grows  at rate $\mu(t)$, which is  strictly greater than $\mu(t) - \alpha(t)$ for a given control $0< \alpha(t) \leq \bar{r}$. We refer to the above discussion, see also, e.g.,  \cite{ECB3,ECB1,ECB4,Jap1,Jap2,Jap3,ECB2}, for more financially oriented ideas  supporting the latter setting. Accordingly, we can assume that the controller aims at maximizing a functional of the following type
\begin{align}
J^{u,a}(t,x) = \mathbb E_t \bigg [\int_t^{\tau \wedge T} &f(X^{u,a}_{t,x}(s),\alpha(s)) ds + g_1(X^{u,a}_{t,x}(T))\Ind{\tau \geq T}+\nonumber\\
& -g_2(X^{u,a}_{t,x}(\tau))\Ind{\tau < T} - \sum_{t \leq \tau_n \leq \tau \wedge T}\left ( K_n+\kappa\right )\bigg ] \; \nonumber .\\
\end{align}
\end{Remark}

In what follows we assume the following density hypothesis on the random time to hold, hence requiring that the distribution of $\tau$ is absolutely continuous with respect to the Lebesgue measure:

\begin{Hypothesis}\label{ASS:D}
For any $t \in [0,T]$, there exists a process $\left (\rho_t(s)\right )_{t \in [0,T]}$, such that
\begin{equation}\label{EQN:DH}
\mathbb{P}\left (\left . \tau \leq s \right |\mathcal{F}_t\right ) = 1-\rho_t(s) \; ,
\end{equation}
\end{Hypothesis}
%{\color{red} This definition seems not to be compatible with the expression \eqref{EQN:GenCPDH}. Do you mean
%$$\mathbb{P}\left(\tau>\,s|\mathcal{F}_t\right) = \rho_t(s)?$$
%or equivalently 
%$$\mathbb{P}\left(\tau\in\,ds|\mathcal{F}_t\right) = \beta(s),$$
%if we use the form 
%\[
%\rho_t(s) := e^{- \int_t^{s} \beta(r)dr}\, .
%\]
%Notice that under this form, we have $\rho_t(s) = \rho_t(\tau)\rho_{\tau}(s)$ for all $t<\tau<s$ and this relation is used in the QVI form that I formulated.
%
%Another possible way to define the density, which is the one that is used in my former paper with Alberto is to assume that
%$$\mathbb{P}\left(\tau\in\,(t,t+ds)|\mathcal{F}_t\right) = \beta(x(t))ds+o(ds),$$
%where $\beta(x)$ is a "bankruptcy risk function" and in this way, we have
%$$\rho_t(s)\doteq\mathbb{P}\left(\tau>s|\mathcal{F}_t\right) =\exp\left\{-\int_t^s \beta(x(t_1))\, dt_1\right\}.$$
%In this way, the running cost will depends on the whole history of $x(s)$ and the under this form, the bank salvage model under the infinite horizon case will have easier QVI form.
%}

The main idea of the following procedure is to switch from the reference filtration $\mathcal{F}_t$, to the default free filtration $\mathcal{G}_t$,
by mean of the following lemma, see \cite[Lemma 4.1.1]{Bie}.

\begin{Lemma}\label{LEM:CF}
For any $\mathcal{F}_T$-measurable random variable $X$ it holds
\begin{equation}\label{EQN:Cox}
\mathbb{E}\left [\left . X \ind_{T \leq \tau}\right |\mathcal{G}_t\right ] =\Ind{\tau > t}\frac{\mathbb{E}\left [\left . X \ind_{\tau >T}\right |\mathcal{F}_t\right ]}{\mathbb{E}\left [\left . \ind_{\tau >t}\right |\mathcal{F}_t\right ]} = \Ind{\tau > t}e^{\Gamma_t} \mathbb{E}\left [\left . X e^{-\Gamma_T}\right | \mathcal{F}_t \right ]\, ,
\end{equation}
with
\[
\Gamma_t := -\ln \left (1-\mathbb{P}\left (\left .\tau \leq t\right |\mathcal{F}_t\right )\right )\,.
\]
\end{Lemma}

%\begin{Remark}\label{REM:CoxPF}
A typical example, which will be used in what follows, consists in considering a Cox process, hence taking  $\rho$ to be an exponential function of the form
\[
\rho_t(s) := e^{- \int_t^{s} \beta(r)dr}\, ,
\]
for a suitable function $\beta$. In this particular case we have that
\[
\Gamma_s := -\ln \left (e^{-\int_t^s \beta(r) dr}\right )= \int_t^s \beta(r)dr\,,
\]
so that the equation \eqref{EQN:Cox} reads
\[
\mathbb{E}\left [\left . X \ind_{T \leq \tau}\right |\mathcal{G}_s\right ]=\Ind{\tau > s}e^{\int_t^s \beta(r)dr} \mathbb{E}\left [\left . X e^{-\int_t^T \beta(r)dr}\right | \mathcal{F}_s \right ]\, .
\]

%\end{Remark}

We can thus prove the following result.

\begin{Hypothesis}\label{ASS:CoxP}
Let us assume that $\tau$ is a Cox process, namely it is of the form
\begin{equation}\label{EQN:CoxP}
\rho_t(s) := e^{- \int_t^{s} \beta(r)dr}\, ,
\end{equation}  
with intensity given by $\beta$.
\end{Hypothesis}

\begin{Remark}
Notice that we could have assumed a more general assumption, often denoted in literature as \textit{density hypothesis}, requiring  that there exists a process $\beta$ such that
\[
\mathbb{P}\left(\tau\in\,ds|\mathcal{F}_t\right) = \beta(s)\, ,
\]
see, e.g. \cite{Bie}.
\end{Remark}

\begin{Theorem}\label{THM:CM}
Let $F$ be a $\mathcal{G}$-adapted process and let us assume $\tau$ to be a Cox process defined as in equation \eqref{EQN:CoxP}, then it holds
\[
\mathbb{E}\left [\left .\int_t^{\tau \wedge T} F_r dr\right |\mathcal{G}_t \right ] = \Ind{\tau > t} \int_t^T \mathbb{E}\left [\left .e^{-\int_t^r \beta(s)ds}  F_r \right | \mathcal{F}_t \right ]dr\, . 
\]
\end{Theorem}
\begin{proof}
Exploiting \eqref{LEM:CF} together with \eqref{EQN:CoxP} we have that
\begin{align}
\mathbb{E}\left [\left .\int_t^{\tau \wedge T} F_r dr\right |\mathcal{G}_t \right ] &= \int_t^{T} \mathbb{E}\left [\left .\Ind{\tau > r}\Ind{\tau > t} F_r\right |\mathcal{G}_t \right ]dr=\nonumber\\
&= \Ind{\tau > t} \int_t^T  \mathbb{E}\left [\left .\Ind{\tau>r} F_r e^{\int_0^t \beta(s)ds}\right | \mathcal{F}_t \right ]dr =\nonumber\\
&=\Ind{\tau > t} \int_t^T e^{\int_0^t \beta(s)ds} \mathbb{E}\left [\left . \mathbb{E}\left [\left .\Ind{\tau>r}\right |\mathcal{F}_r\right ] F_r \right | \mathcal{F}_t \right ]dr=\nonumber\\
&=\Ind{\tau > t} \int_t^T e^{\int_0^t \beta(s)ds} \mathbb{E}\left [\left . e^{-\int_0^r \beta(s)ds} F_r \right | \mathcal{F}_t \right ]dr=\nonumber\\
&=\Ind{\tau > t} \int_t^T \mathbb{E}\left [\left .e^{-\int_t^r \beta(s)ds} F_r \right | \mathcal{F}_t \right ]dr\nonumber\, ,
\end{align}
and this completes the proof.
\end{proof}

Let us then denote the impulse control for this system by
\[
u = (\tau_1,\tau_2,...,\tau_j,...;K_1,K_2,...,K_j,...)\in \mathcal{U}\,,
\]
where $0\leq\tau_1\leq\tau_2\leq...$ are $\mathcal{G}_t$ stopping times and $K_j\in \mathcal{A}$ is $\mathcal{G}_{\tau_j}$- measurable for all j, for any $u\in\mathcal{U}$, then, using \eqref{ASS:CoxP} together with \eqref{THM:CM}, the corresponding functional in equation \eqref{EQN:GenCP} can be rewritten as
\begin{align}
J^{u}(t,x) = \mathbb{E}_t \bigg [\int_t^{T} &\rho_t(s) \Big(f(X(s))-\beta(s)g_2(X(s))\Big)ds +\rho_t(T)g_1(X(T))+\label{EQN:OBJ}\\
&- \sum_{t \leq \tau_n \leq T} \rho_t(\tau_n)\left ( K_n+\kappa\right )\bigg ]\nonumber \,,
\end{align}
so that the original stochastic control problem, with random terminal time, turns out to be a stochastic control problem with deterministic terminal time.

\begin{Remark}
A different approach would be to consider $\tau$ to be $\mathcal{F}_t$--adapted, for instance of the form
\[
\tau = \inf \{ t \,:\, x(t) \leq 0\}\,,
\]
which implies that the hypothesis \eqref{EQN:DH} is no longer satisfied and, consequently, the  above mentioned techniques cannot be exploited any longer. Nevertheless, under this setting it is possible to recover a HJB equation endowed with suitable boundary conditions. We refer to \cite{FS,OS}, for a mathematical treatment of this type of stochastic control problems, while in \cite{Lip,Mer} one can find applications to the mathematical finance scenario.
\end{Remark}

\begin{Theorem}(Dynamic programming principle)\label{THM:DPP}
Let $(t,x) \in [0,T]\times \RR$, then it holds
\begin{align}
V(t,x) = \sup_{u\in\mathcal{U}[t,T]} \mathbb{E}\bigg [\int_t^{\theta} &\rho_t(s) \Big(f(X(s))-\beta(s)g_2(X(s))\Big)ds+\nonumber\\
& -\sum_{t \leq \tau_n \leq \theta} \rho_t(\tau_n)\left ( K_n+\kappa\right ) + \rho_t(\theta)V(\theta,X^u_{t,x}(\theta))\bigg ]\,,\nonumber
\end{align}

for any stopping time $\theta$ valued in $[t,T]$.
\end{Theorem}
\begin{proof}
See, e.g., \cite{OS,Pha}.
\end{proof}
For simplicity, we define the following functions
\begin{equation}\label{EQN:para}
	\begin{array}{l}
c(t,s,x) = \rho_t(s)(f(x)-\beta(s)g_2(x))\quad \hbox{with}\quad s\geq t,\\[3mm]
g(t,x) = \rho_t(T)g_1(x),
\end{array}
\end{equation}
which will be used throughout the paper.

\section{On the regularity of the value function}\label{SEC:Reg}

The present section is devoted to prove regularity properties of the value function. In particular the next two Lemmas prove respectively that the value function is bounded, Lipschitz continuity in space and $\frac{1}{2}-$H\"{o}lder continuity in time of the value function $V$.
\begin{Lemma}\label{l:B}
Let us assume that \eqref{ASS:Lpz} holds, then there exist constants $C_0,C_1$ such that
$$C_1 \geq V(t,x)\geq -C_0(1+|x|).$$
\end{Lemma} 
\begin{proof}
	For simplicity, in what follows, for any fixed $(t,x)\in[0,T]\times\R$ and $u\in {\cal U}[t,T]$, we will denote for short $X^{u}_{t,x}(s)$, resp. $\xi_t(s)$ by $X(s)$, resp. $\xi(s)$. Then by Gronwall's inequality we have
$$
\begin{array}{rl}
1+|X(s)|&\leq\ds 1+|x|+|\xi(s)|	+\left|\int_t^s\sigma(r,X(r))dW_r\right|+C\int_t^s(1+|X(r)|)\,dr\\[3mm]
&\ds\leq 1+|x|+|\xi(s)|	+\left|\int_t^s\sigma(r,X(r))dW_r\right|+\\
&\qquad+C\int_t^s e^{C(s-r)}\left( 1+|x|+|\xi(r)|+\left|\int_t^r\sigma(r,X(r))dW_r\right|\right)dr\\[3mm]
&\ds\leq C\bigg[1+|x|+|\xi(s)|+\int_t^s|\xi(r)|dr+\\
&+\qquad\left|\int_t^s\sigma(r,X(r))dW_r\right|+\int_t^s\left|\int_t^r\sigma(r,X(r))dW_r\right|dr\bigg], 
\end{array}
$$
thus
\begin{align}
\mathbb{E}|X(s)|\leq C\bigg\{&1+\mathbb{E}|x|+\mathbb{E}|\xi(s)|+\mathbb{E}\int_t^s|\xi(r)|dr+\nonumber\\
&+ \mathbb{E}\left|\int_t^s\sigma(r,X(r))dW_r\right|+\mathbb{E}\int_t^s\left|\int_t^r\sigma(r,X(r))dW_r\right|dr\bigg]\,.
\end{align}
On the Other hand, under \eqref{ASS:Lpz}, we have
\begin{equation}\label{EQN:boundvol}
	\begin{array}{rl}
	&\ds \mathbb{E}\left|\int_t^s\sigma(r,X(r))dW_r\right|+\mathbb{E}\left[\int_t^s\left|\int_t^r\sigma(r,X(r))dW_r\right|dr\right]\\[3mm]
	\leq &\ds \left(\mathbb{E}\left|\int_t^s\sigma(r,X(r))\, dWr\right|^2\right)^{1/2}+(s-t)^{1/2}\left(\int_t^s\mathbb{E}\left|\int_t^r\sigma(r,X(r))dW_r\right|^2dr\right)^{1/2}\\[3mm]
	= &\ds \left(\int_t^s \mathbb{E}|\sigma(r,X(r))|^2\, dr\right)^{1/2}+(s-t)^{1/2}\left(\int_t^s\int_t^r \mathbb{E}|\sigma(r,X(r))|^2drdr\right)^{1/2}\\[3mm]
	\leq &\ds \left[1+(s-t)\right]\left(\int_t^s \mathbb{E}|\sigma(r,X(r))|^2\,dr\right)^{1/2}\\[3mm]
	\leq &\ds C\left\{(s-t)^{1/2}+\left(\int_t^s \mathbb{E}|X(r)|^2dr\right)^{1/2}\right\}\leq C\left\{1+\int_t^s \mathbb{E}|X(r)|dr \right\},
\end{array}
\end{equation}
where we have exploited both the Jensen's and H\"{o}lder's inequality, several times.
Hence it follows that
$$
\begin{array}{rl}
\mathbb{E}|X(s)|\leq &
\ds C\left\{1+\mathbb{E}|x|+\mathbb{E}|\xi(s)|+\mathbb{E}\int_t^s|\xi(r)|dr+\int_t^s\mathbb{E}|X(r)|dr\right\}\\[3mm]
\leq &\ds C\left\{1+\mathbb{E}|x|+\mathbb{E}|\xi(s)|+\mathbb{E}\int_t^s|\xi(r)|dr\right\}.	
\end{array}.
$$
Again under \eqref{ASS:Lpz}, we achieve that
$$
\begin{array}{rl}
	|J^{u}(t,x)|&\ds\leq \int_t^{T} C(1+\mathbb{E}|X(s)|)\,ds+C(1+\mathbb{E}|X(T)|)-\sum_{t \leq r_n \leq T} \rho_t(r_n)\left ( K_n+\kappa\right )\\[3mm]
	&\ds\leq C\left(1+|x|+\mathbb{E}|\xi(T)|+\mathbb{E}\int_t^T|\xi(r)|dr\right)
\end{array}
$$
For the trivial control $u_0 = \xi_t(.)\equiv 0 $, one has that
\begin{equation}\label{EQN:Vmin}
	V(t,x)\geq J^{u_0}\geq -C_0(1+|x|)\quad\hbox{for all} (t,x)\in[0,T]\times\\R.
\end{equation}
which proves the lower bound of the value function.
%
%For the upper bound, one has that for every admissible control $\xi(.)\in \mathcal{U}[t,T]$, one has
%$$
%\begin{array}{rl}
%	J^{\xi}(t,x)&\ds\leq \int_t^{T} C\rho_t(s)(1+E|X(s)|)\,ds+C(1+E|X(T)|)-\sum_{t \leq r_n \leq T} \rho_t(r_n)\left ( K_n+\kappa\right )\\[3mm]
%	&\ds\leq C(1+|x|)+\left(C\int_t^T \rho_t(s)\left[E|\xi(s)|+E\int_t^s|\xi(r)|dr\right]-\sum_{t \leq r_n \leq T} \rho_t(r_n)\left ( K_n+\kappa\right )\right).
%\end{array}
%$$
%It seems that it's hard to make any conclusion about the boundedness of 
%$$
%\left(\int_t^T C\rho_t(s)\left[E|\xi(s)|+E\int_t^s|\xi(r)|dr\right]-\sum_{t \leq r_n \leq T} \rho_t(r_n)\left ( K_n+\kappa\right )\right).
%$$
%One may assume that this is finite, but the admissible set should be defined more carefully.

The boundedness of $c(t,s,x),g(t,x)$, immediately gives us that value function is bounded, i.e. there exists $C_1>0$ such that
\begin{equation}\label{EQN:Vmax}
V(t,x)\leq C_1.
\end{equation}
\end{proof}
\begin{Lemma}\label{l:CSHT}
If \eqref{ASS:Lpz} holds, the value function $V(t,x)$ is Lipschitz continuous in $x$, and $\frac{1}{2}$--H\"{o}lder continuous in $t$, namley there exists a constant $C>0$ such that, $\forall$ $t_1\,,\,t_2 \in [0,T)$, $x_1\,,\,x_2\,\in \RR$,
\[
|V(t_1,x_1) -V(t_2,x_2)| \leq C \left (|x_1 - x_2| + (1+|x_1|+|x_2|)|t_1 - t_2|^\frac{1}{2}\right )\,,\quad 
\]
\end{Lemma}
\begin{proof}
Again, for simplicity, for any admissible control $u\in\mathcal{U}[t,T]$, we denote for short $X^{u}_{t,x_1}$, resp $X^u_{t,x_2}$ by $X_{t,x_1}$, resp $X_{t,x_2}$ dropping the explicit dependence on the control $u$. Notice that, applying the It\^{o}-D\"{o}blin formula to $|X_{t,x_1}(s) -X_{t,x_2}(s)|^2$, and using Gronwall's lemma, we can infer that
\[
\mathbb{E}|X_{t,x_1}(s) -X_{t,x_2}(s)|\leq C |x_1 - x_2|\,,\quad \forall \, s \in [t,T]\,,\,x_1\,,\,x_2\,\in \RR\,.
\]

Therefore, by \eqref{ASS:Lpz}, for any fixed $t\in[0,T)$ and all $x_1,x_2\in \mathbb{R}$ and $u\in\mathcal{U}[t,T]$,
%\begin{equation}
\begin{align}
|J^u(t,x_1) - J^u(t,x_2)| &\leq \mathbb{E}\int_t^T |c(t,s,X_{t,x_1}(s)) - c(t,s,X_{t,x_2}(s))| ds + \label{EQN:CS}\\
&+ |g(t,X_{t,x_1}(T)) - g(t,X_{t,x_2}(T))| \\
&\leq L \mathbb{E}\int_t^T |X_{t,x_1}(s) - X_{t,x_2}(s)| ds + C|X_{t,x_1}(T) - X_{t,x_2}(T)| \\
&\leq C |x_1 - x_2|\nonumber\,,
\end{align}
%\end{equation}
%\[
%\mathbb{E}|X_{t_1,x}(s) -X_{t_2,x}(s)|\leq C |t_1 - t_2|^\frac{1}{2}\,,\quad \forall \, t_1\,,\,t_2 \in [0,T)\,,\,x\,\in \RR\;,
%\]
%therefore, by assumptions \eqref{ASS:Lpz}, we obtain
%\begin{equation}\label{EQN:HT}
%\begin{split}
%|J^{\xi_t}(t_1,x) - J^{\xi_t}(t_2,x)| &\leq C |t_1 - t_2|^\frac{1}{2}\, ,\quad \forall \, t_1\,,\,t_2 \in [0,T)\,,\,x\,\in \RR\;,
%\end{split}
%\end{equation}
%and exploiting equations \eqref{EQN:CS}--\eqref{EQN:HT}, we derive the claim.
which implies that
$$V(t,x_1)\leq J^u(t,x_1)\leq J^u(t,x_2)+C|x_1-x_2|,$$
and thus
$$V(t,x_1)\leq V(t,x_2)+C|x_1-x_2|.$$
By interchanging $x_1$ and $x_2$, we get
$$\left|V(t,x_1)- V(t,x_2)\right|\leq C|x_1-x_2|.$$
For the time regularity, first we show that
\begin{equation}\label{EQN:diff1}
	\mathbb{E}\left|X_{t,x}(s)-x-\xi_t(s)\right|\leq C\Big((1+|x|)(s-t)^{\frac{1}{2}}+\mathbb{E}\left(\int_t^s|\xi_t(s)\,ds|\right)\Big).
\end{equation}
For notation simplicity, we suppress the subscripts t,x for $X_{t,x},\,\xi_t$ and define
$$z(s) = X(s)-x-\xi(s).$$
Then by \eqref{ASS:Lpz}, we have
\begin{align}
|z(s)|&\leq C\int_t^s(1+|X(r)|)\,dr+\left|\int_t^s\sigma(r,X(r))\right|\leq \nonumber \\
&\leq C\int_t^s(1+|x|+|z(r)|+|\xi(r)|)\,dr+\left|\int_t^s\sigma(r,X(r))dW_r\right| \nonumber.
\end{align}

By Gronwall's inequality, we achieve
\begin{align}
|z(s)|\leq C\bigg[ &(1+|x|)(s-t)+\int_s^t|\xi(r)|dr+\nonumber\\
&+\left|\int_t^s\sigma(r,X(r))\,dW_r\right|+\int_t^s\left|\int_t^r\sigma(r,X(r))dW_r\right|dr\bigg]\nonumber
\end{align}
Using \eqref{EQN:boundvol} and again Gronwall's, we further get
$$
\begin{array}{rl}
	\mathbb{E}|z(s)|&\ds\leq~ C\left[(1+|x|)(s-t)+\int_s^t\mathbb{E}|\xi(r)|dr+(s-t)^{\frac{1}{2}}+\int_t^s\mathbb{E}|X(r)|\,dr\right]\\[3mm]
	&\ds\leq~C\bigg[(1+|x|)(s-t)+\int_s^t\mathbb{E}|\xi(r)|dr+(s-t)^{\frac{1}{2}}+\\[3mm]
	& ~~~~\ds+  \int_t^s \mathbb{E}(1+|x|+|\xi(r)|+|z(r)|)\,dr\bigg]\\[3mm]
	&\ds\leq~C\left[(1+|x|)(s-t)^{\frac{1}{2}}+\int_s^t\mathbb{E}|\xi(r)|dr\right],
\end{array}
$$
which proves \eqref{EQN:diff1}.\\
For all $p\in[0,\infty)$, define the control space
$$\mathcal{U}_p[t,T] = \left\{u\in\mathcal{U}[t,T]\Bigg\vert \mathbb{E}\sum_{t\leq r_i<T}(\rho_t(r_i)K_i+\kappa)\leq 2C_0(1+p)+C_1\right\},$$
where $C_0$ and $C_1$ are the constants in \eqref{EQN:Vmax} and \eqref{EQN:Vmin}. Notice that another important corollary of \eqref{EQN:diff1} is that for all $u\in\mathcal{U}_{|x|}[t,T]$,
\begin{equation}\label{EQN:diff2}
\mathbb{E}\left|X_{t,x}(s)-x-\xi_t(s)\right|\leq C\Big((1+|x|)(s-t)^{\frac{1}{2}}\Big) ~~\forall~t\leq s\leq T.
\end{equation}
We claim that for all $|x|\leq p$, the value function $V(t,x)$ satisfies
$$V(t,x) = \inf_{u\in\mathcal{U}_p[t,T]} J^u(t,x).$$
This is due to the fact that for any $u\in\mathcal{U}[t,T]\backslash\mathcal{U}_p[t,T]$,
$$J^u(t,x)\leq C_1-\mathbb{E}\sum_{t\leq r_i<T}(\rho_t(r_i)K_i+\kappa)\leq C_1-2C_0(1+p)-C_1<V(t,x)-C_0(1+p).$$

Fix $x\in\mathbb{R}$ and $0\leq t_1<t_2<T$. For any $u_2\in\mathcal{U}_{|x|}[t_2,T)$, extend the control to $[t_1,T)$ by setting 
$$\left\{\begin{array}{l}
\tilde \xi_{t_1}(s) = 0\quad\forall ~s\in[t_1,t_2),\\[2mm]
\tilde \xi_{t_1}(s) = \xi_{t_2}(s)\quad \forall~ s\in[t_2,T).	
\end{array}\right.
$$
and call $\tilde u_1 \doteq \tilde \xi_{t_1}(\cdot)\in\mathcal{U}[t_1,T]$. Then we have
\begin{equation}\label{EQN:holder1}
	\begin{array}{rl}
	V(t_1,x)&\ds\leq~ J^{\tilde u_1}(t_1,x)\\[2mm]
	&\ds = ~J^{u_2}(t_2,x) +\mathbb{E}\int_{t_1}^{t_2}c(t_1,s,X_{t_1,x}(s))\,ds\\[2mm]
	&~~~~\ds+~\mathbb{E}\int_{t_2}^T\left[c(t_1,s,X_{t_1,x}(s))-c(t_2,s,X_{t_2,x}(s))\right]\,ds+\\[2mm]
	&~~~~\ds+~\mathbb{E}\left[g(t_1,X_{t_1,x}(T))-g(t_2,X_{t_2,x}(T))\right]\\[2mm]
	&\ds\leq~J^{u_2}(t_2,x)+C(1+|x|)|t_1-t_2|+ C(1+|x|)(|t_1-t_2|^{\frac{1}{2}})\\[2mm]
	&\ds\leq~ J^{u_2}(t_2,x)+ C(1+|x|)(|t_1-t_2|^{\frac{1}{2}}),
\end{array}
\end{equation}
where $X_{t_1,x}(s)$, resp $X_{t_2,x}(s)$ represents $X^{\tilde u_1}_{t_1,x}$, resp $X^{u_2}_{t_2,x}$ and the second last row in \eqref{EQN:holder1} is achieved by exploiting \eqref{EQN:diff1}. So we obtain that
$$V(t_1,x)\leq V(t_2,x)+C(1+|x|)|t_1-t_2|^{\frac{1}{2}}.$$
On the other hand, for any $\ve>0$, there exists $u_1\in\mathcal{U}_{|x|}[t_1,T)$, such that
$$
\ve+V(t_1,x)\geq J^{u_1}(t_1,x).
$$
Then we define the impulse controls $\hat u_2,\bar u_2\in\mathcal{U}[t_2,T)$ by
$$
\hat \xi_{t_2}(s) = \xi_{t_1}(s)~\forall~s\geq t_2,
$$
$$
\bar \xi_{t_2}(s) = \xi_{t_1}(s)-\xi_{t_1}(t_2)~\forall~s\geq t_2.
$$
Notice that $\hat u_2$ is the impulse control such that at the initial time $t_2$, there is a impulse of size $\xi_{t_1}(t_2)$ and $\bar u_2$ is the impulse control mimicing all the impulses in $\xi_{t_1}(\cdot)$ on $[t_2,T)$.
By denoting $\bar x = x+\xi_{t_1}(t_2)$, which is $\mathcal{F}_{t_2}$ adapted, we have that 
$$J^{\hat u_2}(t_2,x) = J^{\bar u_2}(t_2,\bar x)-(\xi_{t_1}(t_2)+\kappa),$$
and thus
\begin{equation}\label{EQN:holder2}
	\begin{array}{rl}
	\ve+V(t_1,x)&\ds\geq~ J^{\hat u_2}(t_2,x)+\mathbb{E}\int_{t_2}^T(c(t_1,s,X_{t_1,x}(s))-c(t_2,s,X_{t_2,\bar x}(s)))\\[3mm]
	&+\mathbb{E}[g(t_1,X_{t_1,x}(T))-g(t_2,X_{t_2,\bar x}(T))]\,ds\\[3mm]
	&\ds+~\sum_{t_1\leq r_i<t_2}(1-\rho_{t_1}(r_i))(K_i+\kappa)\\[3mm]
	&+\sum_{ r_i\geq t_2}(\rho_{t_2}(r_i)-\rho_{t_1}(r_i))(K_i+\kappa)\\[3mm]
	&\ds\geq~V(t_2,x)-C(1+|x|)(t_2-t_1) +\\[3mm]
	&-C \mathbb{E}|X_{t_1,x}(T)-X_{t_2,\bar x}(T)|- C \int_{t_2}^T\mathbb{E}|X_{t_1,x}(s)-X_{t_2,\bar x}(s)|\,ds\\[3mm]
	&\ds\geq~V(t_2,x) - C(1+|x|)(t_2-t_1)^{\frac{1}{2}},
\end{array}
\end{equation}
where $X_{t_1,x}$, resp $X_{t_2,x}$ represents $X^{u_1}_{t_1,x}$, resp $X^{\bar u_2}_{t_2,x}$. Notice that in \eqref{EQN:holder2},   we extensively use the following inequality
$$
\begin{array}{rl}
	\mathbb{E}|X_{t_1,x}(s)-X_{t_2,\bar x}(s)|&\ds\leq~ C\mathbb{E}|X_{t_1,x}(t_2)-X_{t_2,\bar x}(t_2)|\\[3mm]
	 &\ds=~ C\mathbb{E}|X_{t_1,x}(t_2)-x-\xi_{t_1}(t_2)|\\[3mm]
	 &\ds\leq~ C(1+|x|)(t_2-t_1)^{\frac{1}{2}}
\end{array}
$$
for all $s\geq t_2$ and $u_1\in\mathbb{U}_{|x|}[t_1,T)$, where the last row is achieved by \eqref{EQN:diff2}.\\
Since \eqref{EQN:holder2} holds for all $\ve>0$, we obtain
$$V(t_1,x)\geq V(t_2,x) - C(1+|x|)(t_2-t_1)^{\frac{1}{2}}.$$
Adding \eqref{EQN:holder1}, we finally get the $\frac{1}{2}-$H\"{o}lder continuity in time, i.e.
$$|V(t_1,x)-V(t_2,x)|\leq C(1+|x|)|t_1-t_2|^{\frac{1}{2}}.$$

\end{proof}

\section{Viscosity solution to the Hamilton--Jacobi--Bellman equation}\label{SEC:Vis}

An application of an {\it ad hoc} dynamic programming principle \eqref{THM:DPP}, see, e.g., \cite{OS, Pha}, leads to the following \textit{quasi--variational inequality} (QVI).

%\begin{equation}\label{EQN:QVI}
%\begin{cases}
%\min\left [-\frac{\partial}{\partial t} V(t,x) - \mathcal{L} V(t,x)- f(x)+\beta(t)(V(t,x)+g_2(x))\,,\, V(t,x) - \mathcal{I} V(t,x)\right ]=0\,,\quad \mbox{ on }\, [0,T) \times \RR \\
%\min \left [V(T,x)-g_1(x),V(T,x)-\mathcal{I} V(T,x)\right ]=0\,,\quad \mbox{ on }\, \{T\} \times \RR \, ,
%\end{cases}\;,
%\end{equation}

{\footnotesize
\begin{equation}\label{EQN:QVIuno}
\begin{cases}
\min\left [-\frac{\partial}{\partial t} V(t,x) - \mathcal{L} V(t,x)- f(x)+\beta(t)(V(t,x)+g_2(x))\,,\, V(t,x) - \mathcal{I} V(t,x)\right ]=0\,,\quad \mbox{ on }\, [0,T) \times \RR, \\[2mm]
V(T,x)=g_1(x)\,,\quad \mbox{ on }\, \{T\} \times \RR \, ,
\end{cases}\;
\end{equation}
}

with $\mathcal{I}$ being the non--local impulse operator defined as
\[
\mathcal{I} V(t,x) := \sup_{K \in \mathcal{A}(t,x)} \left [V(t,x+K) - (K+\kappa))\right ]\,.
\]

We underline that the  problem \eqref{EQN:QVIuno} identifies two distinct regions: the \textit{continuation region}
\[
\mathcal{C}= \left \{(t,x)\in [0,T) \times \RR \,:\,V(t,x) >\mathcal{I}V(t,x)\right \}\, ,
\]
and the \textit{impulse region} or \textit{action region}
\[
\mathcal{A}= \left \{(t,x)\in [0,T) \times \RR \,:\,V(t,x) =\mathcal{I}V(t,x)\right \}\, .
\]
Let us consider the following function space.
\begin{Definition}
(Space of polynomially bounded functions).\\
${\mathcal{PB} = \mathcal{PB}([0,T]\times \RR)}$ is the space of all measurable function $u:[0:T]\times \RR\to \RR$ such that
$$|u(t,x)|\leq C_u(1+|x|^p)$$
for some constant $p>0$ and $C_u>0$, independent of $t,x$.
\end{Definition}

Let us introduce in what follows the definition of viscosity solution to the QVI, see eq. \eqref{EQN:QVIuno}, within in the general setting (possibly not continuous).

\begin{Definition}\label{DEF:Viscosity}
A function $V\in\mathcal{PB}$ is said to be a viscosity solution to the QVI \eqref{EQN:QVIuno} if the following two properties hold:
\begin{description}
\item[(i) \textbf{viscosity supersolution}] a function $V\in\mathcal{PB}$ is said to be a \textit{viscosity supersolution} to the QVI \eqref{EQN:QVIuno} if $\forall$ $(\hat{t},\hat{x}) \in [0,T] \times \RR$ and $\phi \in C^{1,2}([0,T]\times \RR)$ with
\[
0=\left (V_* - \phi\right )(\hat{t},\hat{x}) = \min_{(t,x)\in[0,T)\times \RR}\left (V_* - \phi\right ) \,,
\]  
it holds

{\footnotesize
\[
\begin{cases}
\min\left [-\frac{\partial}{\partial t} \phi(t,x) - \mathcal{L} \phi(t,x)- f(x)+\beta(t)(\phi(t,x)+g_2(x))\,,\, V_*(t,x) - \mathcal{I} V_*(t,x)\right ]\geq0\,,\quad \mbox{ on }\, [0,T) \times \RR \\[3mm]
\min \left [V_*(T,x)-g_1(x),V_*(T,x)-\mathcal{I} V_*(T,x)\right ]\geq 0\,,\quad \mbox{ on }\, \{T\} \times \RR \, ,
\end{cases}\;;
\]
}

\item[(ii) \textbf{viscosity subsolution}] a function $V\in\mathcal{PB}$ is said to be a \textit{viscosity subsolution} to the QVI \eqref{EQN:QVIuno} if $\forall$ $(\hat{t},\hat{x}) \in [0,T] \times \RR$ and $\phi \in C^{1,2}([0,T]\times \RR)$ with
\[
0=\left (V^* - \phi\right )(\hat{t},\hat{x}) = \max_{(t,x)\in[0,T)\times \RR}\left (V^* - \phi\right ) \,,
\]  
it holds

{\footnotesize
\[
\begin{cases}
\min\left [-\frac{\partial}{\partial t} \phi(t,x) - \mathcal{L} \phi(t,x)- f(x)+\beta(t)(\phi(t,x)+g_2(x))\,,\, V^*(t,x) - \mathcal{I} V^*(t,x)\right ]\leq0\,,\quad \mbox{ on }\, [0,T) \times \RR \\[3mm]
\min \left [V^*(T,x)-g_1(x),V^*(T,x)-\mathcal{I} V^*(T,x)\right ]\leq 0\,,\quad \mbox{ on }\, \{T\} \times \RR \, ,
\end{cases}\;;
\]
}
\item[(iii) \textbf{viscosity solution}] a function $V\in\mathcal{PB}$ is said to be a \textit{viscosity solution} to the QVI \eqref{EQN:QVIuno} if it is both a \textit{viscosity supersolution} and a \textit{viscosity subsolution}.
\end{description}
\end{Definition}
In order to prove  that the value function $V$ is the viscosity solution to equation \eqref{EQN:QVIuno}, we first need the
following 

\begin{Lemma}\label{LEM:Ob}
Let \eqref{ASS:Lpz} holds, then we have 
\[
V(t,x)\geq \mathcal{I}V(t,x)\,,
\]
for all $t \in [0,T)$, $x \in \RR$.
\end{Lemma}
\begin{proof}
Reasoning by contradiction, we first suppose that there exists $(t,x)\in \mathcal{S}:= [0,T)\times[0,+\infty)$, such that
$$V(t,x) < \mathcal{I}V(t,x),$$
i.e.,
$$V(t,x) < \sup_{K\in\mathcal{A}}V(t,x+K)-(K+k)\,,$$
then there exists also $\epsilon>0$ and $\hat K\in \mathcal{A}$, such that
$$V(t,x) < V(t,x+\hat K) -(\hat K+k)-2\epsilon\;.$$
On the other hand, according to equation \eqref{EQN:val}, there exists $u\in\mathcal{U}[t,T]$ such that
\[
J^u(t,x+\hat K)> V(t,x+\hat K)-\epsilon\,.
\]
Defining now $\hat u = \hat{\xi}_t(\cdot) \doteq \hat{K} + \xi_t(\cdot)$, we have
$$V(t,x)\geq J^{\hat u}(t,x) = J^{u}(t,x+\hat K)-(\hat K+k)\;.$$
Combining all the estimates above, we have
$$V(t,x+\hat K)-(\hat K+k)-2\epsilon>V(t,x)>V(t,x+\hat K)-(\hat K+k)-\epsilon\;,$$
from which we have the desired contradiction.
\end{proof}

\begin{Remark}
\eqref{LEM:Ob} implies that we are considering $\mathcal{I}V(t,x)$ as a lower obstacle, which is given in implicit form, since it depends on the value function $V$ itself.
\end{Remark}

\begin{Theorem}\label{THM:E}
The value function $V(t,x)$ is a {\it viscosity solution} to the QVI \eqref{EQN:QVIuno} on $[0,T] \times \RR$, in the sense of \eqref{DEF:Viscosity}.
\end{Theorem}
\begin{proof}
\eqref{l:CSHT} implies that the value function is continuous. Therefore the lower--semicontinuous, resp. upper--semicontinuous, envelop of $V$ in \eqref{DEF:Viscosity}, does in fact coincide with $V$. Also, it is an immediate consequence of \eqref{l:B} that $V\in\mathcal{PB}.$

Let us prove that $V(t,x)$ is a viscosity sub-solution of \eqref{EQN:OBJ}. By \eqref{LEM:Ob},  we know that $V(t,x)\geq \mathcal{I}V(t,x)$, so that in what follows we only need to show that given $(t_0,x_0) \in [0,T) \times \RR$ such that
\begin{equation}\label{EQN:SUBassump}
V(t_0,x_0)> \mathcal{I}V(t_0,x_0)\;,
\end{equation}
then for every $\phi(t,x)\in \mathcal{C}^{1,2}([0,T]\times[0,+\infty))$ and every $t_0,x_0\in[0,+\infty)$, such that $\phi\geq V$ for all $(t,x)\in \mathcal{S}\cap B_r((t_0,x_0))$ and $V(t_0,x_0) = \phi(t_0,x_0)$, we want to show that
\begin{equation}\label{EQN:SUB}
-\frac{\partial}{\partial t} \phi(t_0,x_0) - \mathcal{L} \phi(t_0,x_0)-f(x_0)+ \beta(t_0) \big(
\phi(x_0)+g_2(x_0)\big)\leq 0\,.
\end{equation}
In fact, if $V(t_0,x_0)\leq \mathcal{I}V(t_0,x_0)$, then \eqref{EQN:SUB} immediately follows.

Choose $\epsilon>0$ and let $u = (\tau_1,\tau_2,...;K_1,K_2,...)\in \mathcal{U}[t_0,T]$ be a $\epsilon$- optimal control, i.e.,
$$V(t_0,x_0)< J^{u}(t_0,x_0) +\epsilon.$$
Since $\tau_1$ is a stopping time, $\{\omega,\tau_1(\omega) = t_0\}$ is $\mathcal{F}_{t_0}$- measurable, thus 
$$\tau_1(\omega) = t_0~ a.s.\quad\quad or\quad\quad \tau_1(\omega)>t_0~a.s.$$
If $\tau_1 = t_0$ a.s., $X^{u}_{t_0,x_0}$ takes a immediate jump from $x_0$ to the point $x_0+K_1$ and we have $J^{u}(t_0,x_0) = J^{u'}(t_0,x_0+K_1)-(K_1+k)$ where $u' = (\tau_2,\tau_3,...;K_2,K_3,...)\in\mathcal{U}[t_0,T]$. This implies that
$$V(t_0,x_0)\leq J^{u'}(t_0,x_0+K_1)-(K_1+k)+\epsilon< V(t_0,x_0+K_1)-(K_1+k)\leq \mathcal{I}V(t_0,x_0)+\epsilon,$$
which is a contradiction for $\epsilon< V(t_0,x_0)-\mathcal{I}V(t_0,x_0)$. Thus \eqref{EQN:SUBassump} implies that $\tau_1>t_0~a.s$ for all $\epsilon$- optimal controls such that $\epsilon< V(t_0,x_0)-\mathcal{I}V(t_0,x_0)$.
For any impulse control $u = (\tau_1,\tau_2,...;K_1,K_2,...)\in\mathcal{U}[t_0,T]$, define 
$$\hat\tau\doteq\tau_1\wedge (t_0+r)\wedge \inf\left\{t\Big|t>t_0,|x(t)-x_0|\geq r\right\}.$$

By the dynamic programming principle, for any $\epsilon>0$, there exists a control $u$ such that
\begin{equation}\label{EQN:SUB1}
	V(t_0,x_0)\leq E^{t_0,x_0}\left[\int_{t_0}^{\hat\tau}\Big\{\rho_{t_0}(s)\big(f(X(s))+\beta(s)g_2(x(s))\big)\Big\} ds+e^{{-\int_{t_0}^{\hat\tau} \beta(s)\,ds}}V(\hat\tau,X(\hat\tau)))\right]~+~\epsilon.
\end{equation}

By \eqref{EQN:SUB1} and the Dynkin formula we have that
\begin{equation}
\label{EQN:SUB2}
\begin{array}{l}
	V(t_0,x_0)~\leq~\ds \\[3mm]
	\leq E^{t_0,x_0}\left[\int_{t_0}^{\hat \tau}\Big\{\rho_{t_0}(s)\big(f(X(s))+\beta(s)g_2(X(s))\big)\Big\}\,ds+e^{-\int_{t_0}^{\hat\tau} \beta(s)\,ds}\phi(\hat\tau,X(\hat\tau)))\right]+\epsilon\\[3mm]
	=~\ds E^{t_0,x_0}\left[\int_{t_0}^{\hat \tau}\Bigg\{\rho_{t_0}(s)\Big(f(X(s))+\beta(s)\big(g_2(X(s))-\phi(s,X(s))\big)+ \Big)\Bigg\}\,ds\right]+\\[3mm]
	~\ds E^{t_0,x_0}\left[\int_{t_0}^{\hat \tau}\Bigg\{\rho_{t_0}(s)\Big(\frac{\partial}{\partial t}\phi(s,X(s))+\mathcal{L}\phi(s,X(s))\Big)\Bigg\}\,ds\right]+\ds \phi(t_0,x_0) +\epsilon.
\end{array}
\end{equation}

Using that $V(t_0,x_0) = \phi(t_0,x_0)$, we further obtain

\begin{align}
	E^{t_0,x_0}\bigg[ \int_{t_0}^{\hat \tau}\Bigg\{&\rho_{t_0}(s)\Big(f(X(s))+\beta(s)\big(g_2(X(s))-\phi(s,X(s))\big)+\label{EQN:SUB4}\\
	&+\frac{\partial}{\partial t}\phi(s,X(s))+\mathcal{L}\phi(s,X(s))\Big)\Bigg\}\,ds\bigg]+ \phi(t_0,x_0)\geq -\epsilon.\nonumber
\end{align}

Divide both sides of \eqref{EQN:SUB4} by $E(\hat\tau-t_0)$ and let $r\rightarrow 0$, we further get that
$$f(x_0)+\beta(t_0)\big(g(x_0)-\phi(t_0,x_0)\big)+\frac{\partial}{\partial t}\phi(t_0,x_0)+\mathcal{L}\phi(t_0,x_0)\geq-\epsilon.$$
Since $\epsilon>0$ is arbitrary, we finally get the desired inequality
\begin{equation}
	\label{EQN:SUB3}
	-f(x_0)+\beta(t_0)\big(\phi(t_0,x_0)+g_2(x_0)\big)-\frac{\partial}{\partial t}\phi(t_0,x_0)-\mathcal{L}\phi(t_0,x_0)\leq 0 \;,
\end{equation}
then $V(t,x)$ is a viscosity sub-solution.

To prove that $V(t,x)$ is also a viscosity super-solution of \eqref{EQN:OBJ}, let us consider $\phi\in \mathcal{C}^{1,2}(\mathcal{S})$, and any $(t_0,x_0)\in \mathcal{S}$ such that $\phi\leq V$ on $B_r(t_0,x_0)$ and $\phi(t_0,x_0) = V(t_0,x_0)$. Taking the  trivial control $u_0 = 0$ ( no interventions ), calling the corresponding trajectory $X(t) = X^{u_0}(t)$ with $x(t_0) = x_0$,  and defining $\hat \tau = (t_0+r)\wedge\inf\left\{t\Big|t>t_0,|x(t)-x_0|>r\right\}$, then, by the dynamic programming principle and the Dynkin formula, we have
$$
\begin{array}{l}
	V(t_0,x_0)~\geq~\ds E^{t_0,x_0}\left[\int_{t_0}^{\hat \tau}\Big\{\rho_{t_0}(s)\big(f(X(s))+\beta(s)g_2(X(s))\big)\Big\}\,ds+e^{-\int_{t_0}^{\hat\tau} \beta(s)\,ds}\phi(\hat\tau,X(\hat\tau)))\right]\\[3mm]
	=~\ds E^{t_0,x_0}\left[\int_{t_0}^{\hat \tau}\Bigg\{\rho_{t_0}(s)\Big(f(x(s))+\beta(s)\big(g_2(X(s))-\phi(s,X(s))\big)\Big)\Bigg\}\,ds\right]+\\[3mm]	
	~\ds E^{t_0,x_0}\left[\int_{t_0}^{\hat \tau}\Bigg\{\rho_{t_0}(s)\Big(\frac{\partial}{\partial t}\phi(s,X(s))+\mathcal{L}\phi(s,X(s))\Big)\Bigg\}\,ds\right]+\ds \phi(t_0,x_0) +\epsilon.
\end{array}
$$ 
Using $V(t_0,x_0) = \phi(t_0,x_0)$, we obtain that 

\begin{align}
	E^{t_0,x_0}\bigg[\int_{t_0}^{\hat \tau}\Bigg\{&\rho_{t_0}(s)\Big(f(X(s))+\beta(s)\big(g_2(X(s))-\phi(s,X(s))\big)+\label{EQN:SUP1}\\
	&\frac{\partial}{\partial t}\phi(s,X(s))+\mathcal{L}\phi(s,X(s))\Big)\Bigg\}\,ds\bigg]\leq-\epsilon.\nonumber
\end{align}

Divide both sides of \eqref{EQN:SUP1} by $E[\hat\tau-t_0]$ and let $r\rightarrow 0$, we obtain 
\begin{equation}
	\label{EQN:SUP2}
	-f(x_0)+\beta(t_0)\big(\phi(t_0,x_0)+g_2(x_0)\big)-\frac{\partial}{\partial t}\phi(t_0,x_0)-\mathcal{L}\phi(t_0,x_0)\geq 0\;.
\end{equation}
Since we have already proved that $V(t,x)\geq \mathcal{I}V(t,x)$, we finally conclude that 
\begin{align}
	\min\bigg [-&\frac{\partial}{\partial t} \phi(t_0,x_0) - \mathcal{L} \phi(t_0,x_0)-f(x_0)+ \beta(t_0) \big(
\phi(x_0)+g_2(x_0)\big)\,,\label{EQN:SUP}\\
\,&, V(t_0,x_0) - \mathcal{I} V(t_0,x_0)\bigg ]\geq 0\;.
\end{align}
Combining \eqref{EQN:SUB} and \eqref{EQN:SUP}, we have that $v(t,x)$ is a viscosity solution of \eqref{EQN:OBJ}.
it is worth to mention that the terminal condition is non trivial. In fact, it has to take into account that just right before the horizon time  $T$, the controller might act by an impulse control. To this extent we have to specify that the terminal condition in equation \eqref{EQN:QVIuno} is to be intended as
\[
V(T,x) := \lim_{(t,x') \to (T^-,x))}V(t,x')\,.
\]
Since \eqref{LEM:Ob} implies that $V(t,x)\geq \I V(t,x)$ for all $(t,x)\in[0,T)\times\R$, in the limit one has $V(T,x)\geq \I V(T,x)$ for all $x\in\R$. To show the boundary condition
\begin{equation}\label{EQN:bnd}
\min\big\{V(T,x)-g_1(x),V(T,x)-\I V(T,x)\big\} = 0,	
\end{equation}
one first consider all the $x\in\R$ such that $V(T,x)>\I V(T,x)$. For any sequence $(t_n,x_n)\rightarrow (T,x)$ with $(t_n,x_n)\in [0,T)\times\R$, by continuity one has $V(t_n,x_n)> \I V(t_n,x_n)$ for all $n$ large enough. Then for each $\ve>0$ sufficiently small, consider controls $u_n\in\mathcal{U}[t_n,T]$ such that
$$ V(t_n,x_n)\leq J^{u_n}(t_n,x_n)+\ve.$$
It then suffices to show that
\begin{align}
\mathbb{E} \left [\int_{t_n}^{T} \rho_{t_n}(s) \Big(f(X^{u_n}_{t_n,x_n}(s))-\beta(s)g_2(X^{u_n}_{t_n,x_n}(s))\Big)ds\right ] +\label{EQN:bndlimit}\\
+\mathbb{E} \left [\rho_{t_n}(T)g_1(X^{u_n}_{t_n,x_n}(T))- \sum_{t_n \leq \tau_j \leq T} \rho_{t_n}(\tau_j)\left ( K_j+\kappa\right )\right ]\rightarrow g_1(x)\nonumber
\end{align}
as $n\rightarrow +\infty$. Notice that since $V(t,x')> \I V(t,x')$ for all $(t,x')$ in a neighborhood of $(T,x)$, for all $\delta>0$ small enough one has
$$\mathbb{P}\Big(\sup_{s\in [t_n,T]}|X^{u_n}_{t_n,x_n}(s)-x_n|<\delta\Big)\rightarrow 1\quad\hbox{as}~n\rightarrow \infty.$$
Suppose that there exists $F\in L^1(\mathbb{P};\R)$ such that

\begin{align}
&\int_{t_n}^{T} \rho_{t_n}(s) \Big(|f(X^{u_n}_{t_n,x_n}(s))|+\beta(s)|g_2(X^{u_n}_{t_n,x_n}(s))|\Big)ds +\nonumber\\
&\rho_{t_n}(T)|g_1(X^{u_n}_{t_n,x_n}(T))|+ |\sum_{t_n \leq \tau_j \leq T} \rho_{t_n}(\tau_j)\left ( K_j+\kappa\right )|\leq F\nonumber
\end{align}

for all $n$ large enough, an application of dominant convergence theorem proves \eqref{EQN:bndlimit}. So we conclude that for any $(T,x)$ such that $V(T,x)>\I V(T,x)$, one has 
$$V(T,x)\leq g_1(x)+\ve,~\forall \ve>0\quad\implies \quad V(T,x)\leq g_1(x).$$
By a similar approach, one can show that $V(T,x)\geq g_1(x)$ for all $x\in\R$. This completes the proof of \eqref{EQN:bnd}.
\end{proof}

We are now to show that the value function is the unique viscosity solution to equation \eqref{EQN:QVIuno} based on a comparison principle. In order to do that let us introduce a different definition of viscosity solution, see, e.g. \cite{Ish}, based on the notion of jets.

\begin{Definition}\label{DEF:Vis}
Let $V:[0,T] \in\mathcal{PB}$ a upper--semicontinuous function, then we define

\begin{align}
\mathcal{P}^{2,+} V(s,x) &= \Big \{ \left (p,q,M\right ) \in \RR \times \RR \times \RR \, :\nonumber \,\\
V(s,y) & \leq V(t,x) + p(t-s) + q (x-y) + \frac{1}{2} M (x-y)^2 + o(|t-s| + |x-y|^2)\Big \}\nonumber\,\\
\bar{\mathcal{P}}^{2,+} V(s,x) &= \left \{ \left (p,q,M\right ) \in \RR \times \RR \times \RR \, : \,\right . \exists (t_n,x_n) \in [0,T]\times \RR\,:  \left (p_n,q_n,M_n\right ) \in \mathcal{P}^{2,+} V(t_n,x_n)\,\nonumber\\
 & \left . \, , \,\left (t_n,x_n,V(t_n,x_n),p_n,q_n,M_n\right ) \to \left (t,x,V(t,x),p,q,M\right )  \right \}\nonumber\,.
\end{align}

For  lower--semicontinuous function $V$, we define
\[
\mathcal{P}^{2,-} V(s,x):= -\mathcal{P}^{2,+} -V(s,x)  \,,\quad \bar{\mathcal{P}}^{2,-} V(s,x):= -\bar{\mathcal{P}}^{2,+} -V(s,x) \,.
\]
\end{Definition}

We can therefore state the equivalence between the two notion of viscosity solution stated before.

\begin{Proposition}
A function $V\in\mathcal{PB}$ is a viscosity sub, resp. super, solution to equation \eqref{EQN:QVIuno} if and only if $\forall$ $(p,q,M) \in \bar{\mathcal{P}}^{2,+} V(s,x)$, resp. $\bar{\mathcal{P}}^{2,-} V(s,x)$,
\begin{align}
\min\bigg [&-p - \mu(t,x) q -\frac{1}{2} \sigma^2(t,x)M- f(x)+\beta(t)(V(t,x)+g_2(x))\,,\\
&,\, V(t,x) - \mathcal{I} V(t,x)\bigg ]\leq 0\, ( \, \geq 0)\,.
\end{align}
\end{Proposition}

\begin{Theorem}[Comparison principle]\label{THM:U}
Suppose that \eqref{ASS:Lpz} is satisfied and that $U$ and $V$ are, repectively, a viscosity super solution and viscority sub solution to the equation \eqref{EQN:QVIuno}. Assume also that $U$ and $V$ are uniformly continuous, then $V \leq U$ on $[0,T]\times \RR$.
\end{Theorem}
\begin{proof}
Let us prove the result by contradiction, assuming that
\[
\sup_{[0,T] \times \RR} (V-U) = \eta >0\,.
\]

For $r >0$ let us define
\[
\tilde{V}(t,x) := e^{rt}V(x,t)\,,\quad \tilde{U}(t,x) := e^{rt}U(t,x)\,.
\]

From the theorem hypotheses, that is $U$ and $V$ are viscosity super and sub solution to equation \eqref{EQN:QVIuno}, we immediately have that $\tilde{V}$ and $\tilde{U}$ are viscosity super and sub solution to
\begin{align}\label{EQN:QVI}
\begin{cases}
\min\bigg [ &r u(t,x) -\frac{\partial}{\partial t} u(t,x) - \mathcal{L} u(t,x)- e^{rt}f(x)+e^{rt}\beta(t)(V(t,x)+g_2(x))\,,\nonumber \\
&,\, u(t,x) - \tilde{\mathcal{I}} u(t,x)\bigg ]=0\,,~~ \mbox{ on }\, [0,T) \times \RR \nonumber\\
u(T,x)&=e^{rt}g_1(x)\,,\quad \mbox{ on }\, \{T\} \times \RR \, ,
\end{cases}\;,
\end{align}
with $\tilde{\mathcal{I}}$ being the non--local impulse operator defined as
\[
\tilde{\mathcal{I}} u(t,x) := \sup_{K \in \mathcal{A}(t,x)} \left [u(t,x+K) - e^{rt}(K+\kappa))\right ]\,.
\]

Let us then assume that for $x_0 \in \RR$ we have that
\[
\tilde{V}(T,x_0) -\tilde{U}(T,x_0)>0\,,
\]
and from the fact that $\tilde{U}$ is a viscosity super solution, resp. $\tilde{V}$ is a viscosity sub solution, we have that it exists $\bar{x}$ such that
\[
\tilde{U}(T,\bar{x}) < \tilde{\mathcal{I}}\tilde{U}(T,\bar{x}) \,,\quad \mbox{resp.} \, \tilde{V}(T,\bar{x}) > \tilde{\mathcal{I}}\tilde{V}(T,\bar{x})\,.
\]

Since we also have that $\tilde{V}(T,\bar{x}) \leq e^{rt}g_1(\bar{x})$ and $\tilde{U}(T,\bar{x}) \geq e^{rt}g_1(\bar{x})$, we conclude that
\[
\tilde{V}(T,\bar{x}) -\tilde{U}(T,\bar{x})\leq0\,,
\]
which contradict the assumptions.

Then suppose that there exists $(\bar{t},\bar{x})\in [0,T)\times \RR$, such that
\[
\tilde{V}(\bar{t},\bar{x}) -\tilde{U}(\bar{t},\bar{x})>0\,,
\]
then, analogously to what we have derived above, we have that
\[
\tilde{U}(t,x) < \tilde{\mathcal{I}}\tilde{U}(t,x) \,,\quad \mbox{resp.} \,\quad \tilde{V}(t,x) > \tilde{\mathcal{I}}\tilde{V}(t,x)\,,
\]
for $t \in I_\delta:=[\bar{t}-\delta,\bar{t}+\delta]$ and $x \in B_\delta:=[\bar{x}-\delta,\bar{x}+\delta]$. 

Therefore taking $(t_0,x_0)\in I_\delta \times B_\delta$, such that
\[
\sup_{I_\delta \times B_\delta} \tilde{V}-\tilde{U}=(\tilde{V}-\tilde{U})(t_0,x_0)>0\,,
\]
and considering
\[
\varphi_n(t,x,y):= \tilde{V}(t,x)-\tilde{U}(t,x) -\varrho_n (t,x,y)\,,
\]
with
\[
\varrho_n (t,x,y)= n|x-y|^2 + |x-x_0|^4 + |t-t_0|^2 \;,
\]
for any $n \in \mathbb{N}$ there exist a point $(t_n,x_n,y_n)$ attaining the maximum of $\varphi$, so that, up to a subsequence, we have
\begin{equation}\label{EQN:ConvUV}
\tilde{V}(t_n,x_n)-\tilde{U}(t_n,x_n) \to \tilde{V}(t_0,x_0)-\tilde{U}(t_0,x_0)\,,\quad \mbox{as}\, n \to \infty\,.
\end{equation}

Moreover, since
\[
\tilde{V}(t_0,x_0)-\tilde{U}(t_0,x_0) = \varphi_n (t_0,x_0,x_0) \leq \varphi_n (t_n,x_n,x_n)\,,
\]
then

\begin{align}
\tilde{V}(t_0,x_0)-\tilde{U}(t_0,x_0) &\leq \liminf_{n \to \infty} \varphi_n (t_0,x_0,y_0)\leq  \limsup_{n \to \infty} \varphi_n (t_0,x_0,y_0)\leq \nonumber\\
&\leq \tilde{V}(\bar{t},\bar{x})-\tilde{U}(\bar{t},\bar{x}) -\liminf_{n \to \infty} n|x-y|^2 + |x-x_0|^4 + |t-t_0|^2\,.\nonumber
\end{align}

Therefore, using the optimality of $(x_0,t_0)$, we obtain that, considering up to a subsequence it holds $(t_n,x_n,y_n)\to (t_0,x_0,x_0)$ and $n|x_n-y_n|\to 0$.

Applying the Ishii lemma, we have that there exists $(p^n_V,q^n_V,M^n_V) \in \bar{\mathcal{P}}^{2,+} \tilde{V}(t_n,x_n)$ and $(p^n_U,q^n_U,M^n_U) \in \bar{\mathcal{P}}^{2,-} \tilde{U}(t_n,x_n)$, such that

\begin{align}
&p^n_V - p^n_V = 2(t_n t_0)\,,\nonumber\\
&q^n_V = \partial_x \varrho_n\,,\quad q^n_U = - \partial_y \varrho_n\,,\nonumber
\end{align}

and
\[
\begin{pmatrix}
M_n & 0 \\
0 & -N_m
\end{pmatrix}
\leq A_n + \frac{1}{2n}A^n_n\,,
\]
with $A_n = \partial_{xy} \varrho_n$. Therefore from the viscosity sub-solution property of $\tilde{V}$, resp. the viscosity super-solution property of $\tilde{U}$,  by the Lipschitz continuity of $\mu$ and $\sigma$ in $x$ and \eqref{DEF:Vis} 
%and using equation 
we have that
\[
r\left (\tilde{V}(t_0,x_0) - \tilde{U}(t_0,x_0)\right ) \leq 0\,,
\]
which gives the desired contradiction.
\end{proof}

We are now able to state the uniqueness result for the viscosity solution.

\begin{Corollary}
Let \eqref{ASS:Lpz} hold true, then there exists a unique viscosity solution to equation \eqref{EQN:QVIuno}.
\end{Corollary}

\begin{proof}
Let $V_1$ and $V_2$ two viscosity solution to equation \eqref{EQN:QVIuno}; then since $V_1$ is a subsolution and $V_2$ is a supersolution, by comparison principle \eqref{THM:U} we obtain that $V_2 \leq V_1$. Since it must also holds the opposite we obtain the claim.
\end{proof}

\section{Smooth fit principle on the value function}\label{SEC:SF}

Under further regularity assumptions on the coefficients, to be further specified in a while, one can prove the regularity property of the value function, with particular reference to the smooth-fit property through the switching boundaries between action and continuation regions. This results, known as \textit{smooth--fit principle}, see, e.g, \cite{ HHY, EgamiYamazaki, GuoWu}, has already been proven to hold in the infinite horizon case. Also, we will prove $W^{(1,2),p}_{loc}$ regularity for the value function $V(t,x)$ on any fixed parabolic domain
$Q_T\doteq (\delta,T]\times B_R(0)$ for any constants $0<\delta<T,~R>0$. 
%Furthermore, we relax the restriction of $\sigma(s,x)$ defined in \eqref{EQN:para} to $\C([0,T]\times\R)$.

In what follows we introduce the definition of the function spaces we are going to use 
throughout the section, $\Omega$ being a bounded open set:
\begin{equation}
	\begin{array}{l}
	W^{(0,1),p}(\Omega) = \{u\in L^p(\Omega): u_{x_i}\in L^p(\Omega)\},\\[3mm]
	W^{(1,2),p}(\Omega) = \{u\in W^{(0,1),p}(\Omega): u_{x_ix_j}\in L^p(\Omega)\},\\[3mm]
	C^{0+\frac{\alpha}{2},0+\alpha}(\bar\Omega) = \left\{u\in C(\bar\Omega):\sup_{(x,t),(y,s)\in\Omega,(x,t)\neq(y,s)}\frac{|u(t,x)-u(s,y)|}{(|t-s|+|x-y|^2)^{\alpha/2}}<+\infty\right\},\\[3mm]
	C^{1+\frac{\alpha}{2},2+\alpha}(\bar\Omega) = \left\{u\in C(\bar\Omega):u_t,u_{x_ix_j}\in C^{0+\frac{\alpha}{2},0+\alpha}\right\},\\[3mm]
	W^{(1,2),p}_{loc}(\Omega) = \left\{u\in L^p_{loc}(\Omega):u\in W^{(1,2),p}(U)~\forall \hbox{ open } U \hbox{ with }\bar U\subset\bar\Omega\backslash\partial_P\Omega\right\}.
	\end{array}
\end{equation}
The above notations are similar to the notations used in \cite{GuoChen}.

Recall that $\beta(t)$ is the hazard rate function defined in \eqref{EQN:CoxP} and we make the following assumption:
\begin{Hypothesis}\label{ASS:smoothfit}
	Let $\alpha\in(0,1]$, we assume that the intensity function $\beta(t)\in C^{\alpha/2}([0,T])$ and $\sigma(s,x)\in \C^{0+\alpha,0+\frac{\alpha}{2}}(\bar Q_T)$ satisfying the uniform elliptic condition, i.e.
	$$\sigma(s,x)\geq \delta_1>0$$
 for some constant $\delta_1>0$ depending on the domain $Q_T$.
\end{Hypothesis}
Before proceeding to the smooth fit principle, recall that we divide the region $[0,T]\times \mathbb{R}$ into the following regions:
$$
\begin{array}{l}
	\mathcal{C} \doteq \left\{(t,x) : V(t,x)>\mathcal{I}V(t,x)\right\},\\[3mm]
	\mathcal{A} \doteq \left\{(t,x) : V(t,x)=\mathcal{I}V(t,x)\right\}
\end{array}
$$
and for any open set $\Omega\in \R^2$, the parabolic boundary $\partial_P\Omega$ is defined as
$$\partial_P\Omega\doteq\left\{(t,x)\in\bar\Omega|~\forall \ve>0,~Q((t,x),\ve) ~\hbox{contains points not in}~\Omega \right\},$$
where $Q((t_0,x_0),r) \doteq \{(t,x);|x-x_0|<r,t<t_0\}$ for all $(t_0,x_0,r)\in\R^2\times\R^+$.
For any $(t,x)\in \A$, define the set
$$\Theta(t,x) = \left\{\xi_0~\vert~\I V(t,x) = V(t,x+\xi_0)-\xi_0-\kappa\right\}.$$ 

Notice that in the regularity analysis in Section 2, we already show that $V(T-t,x)\in \C^{0+1/2,0+1}(\Omega)$, so we immediately have the following lemma.
\begin{Lemma}\label{l:41}(Theorem 4.9,5.9,5.10,and 6.33 in \cite{GML})
	Under \eqref{ASS:Lpz} and \eqref{ASS:smoothfit}, for any open set $\Omega\subseteq \C$, the linear parabolic PDE
\begin{equation}\label{EQN:parabolic}
\left\{
\begin{array}{l}
	u_t- \mathcal{L} u(t,x) +\tilde \beta(t)u(t,x)= {\tilde f}(t,x)\quad \hbox{in}~ \Omega,\\[3mm]
	u(t,x) = V(T-t,x),\quad\hbox{on}~\partial_P\Omega.
\end{array}\right.
\end{equation}
admits a unique solution $u(t)\in \C^{0+\alpha/2,0+\alpha}(\bar \Omega)\cap\C^{1+\alpha/2,2+\alpha}_{loc}(\Omega)$ where 
$$\tilde \beta(t) = \beta(T-t),\quad\tilde f(t,x) = f(x)-\beta(t)g_2(x).$$
\end{Lemma}

\begin{Theorem}\label{Thm:42}(Smooth fit principle)
	Under \eqref{ASS:Lpz} and \eqref{ASS:smoothfit}, the value function $V(t,x)$ is a unique $W^{(1,2),p}_{loc}(\R\times(0,T))$ viscosity solution to the QVI \eqref{EQN:QVIuno} for any $1<p<+\infty$. Furthermore, for any $t\in [0,T)$, $V(t,\cdot)\in \C^{1,\gamma}_{loc}(\R)$ for any $0<\gamma<1$.
\end{Theorem}
\begin{proof}
Using the cost function 
$$B(K)\doteq K+\kappa,~\forall~K>0,$$
which is independent of time and satisfies the subadditivity property, i.e.
\begin{equation}\label{subadd}
	B(K_1+K_2)+\kappa= B(K_1)+B(K_2),~\forall~ K_1,K_2>0.
\end{equation}
so that the claim follows from \cite{GuoChen} together with \eqref{l:41}.
% \eqref{EQN:parabolic} contains an additional linear term $\beta(t)V(t,x)$, but with \eqref{l:41} the proof is exactly the same. 
\end{proof}

\subsection{Structure of the value function}
In this subsection, we study the general property of the value function $V(t,x)$ under further assumptions of $\sigma(t,x)$, $\beta(t)$, $\mu(t,x)$, $\tilde f(t,x)$ and  $g_1(x)$. 
\begin{Hypothesis}\label{ASS:monotone}
	$\tilde f(t,x)$ and $g_1(x)$ are monotonically increasing with
$$\lim_{x\rightarrow -\infty}\tilde f(t,x)= \lim_{x\rightarrow -\infty}g_1(x)= -\infty,\qquad \lim_{x\rightarrow +\infty}\tilde f(t,x) =  U(t)>0, \lim_{x\rightarrow+\infty}g_1(x) =  U_g<\infty.$$
\end{Hypothesis}

\begin{Lemma}
	Under \eqref{ASS:monotone}, for any $t>0$ the value function $V(t,x)$ satisfies
	$$V(t,x_1)\leq V(t,x_2)~\forall~ x_1\leq x_2.$$
	Furthermore, there exists $L\in[-\infty,+\infty)$ such that
	$$[0,T]\times(L,+\infty)\subset\C.$$
\end{Lemma}
\begin{proof}
	First, we show the monotonicity of $V(t,x)$ with respect to x. By applying the same adapted control $u\in\mathcal{U}[t,T]$ with different initial values $x_1\leq x_2$, the solutions satisfies $X^u_{t,x_1}\leq X^u_{t,x_2}~a.s$. Since $\tilde f(t,x)$ is increasing with respect to x, one has $J^u(t,x_1)\leq J^u(t,x_2)$ for all $u\in\mathcal{U}[t,T],$ and thus $V(t,x_1)\leq V(t,x_2)$ for any $x_1\leq x_2$.
	
	It remains to show that there exists $L\in[-\infty,+\infty)$ such that for any fixed $t>0$ and any $x_0>L$, $(t,x_0)\in\C$. Fix any $t\in(0,T)$, suppose that there exists a sequence $x_1<x_2<...<x_k<...$ such that 
	$$\lim_{k\rightarrow+\infty}x_k = +\infty~\hbox{and}~(t,x_k)\in\A, ~\forall~ k>0,$$ 
	and for any $k>0$ there exists $\xi_k\in\Theta(t,x_k)$ such that
	\begin{equation}\label{Vxk}
		V(t,x_k) = V(t,x_k+\xi_k)-\xi_k-\kappa.
	\end{equation}
	However, since $V(t,x)$ is monotone, uniformly Lipschitz continuous in x and upper bounded by $C_1$ according to \eqref{l:B} and \eqref{l:CSHT}, for any $\ve>0$ one can choose $L$ large enough such that
	$$V(t,x+\xi)-V(t,x)\leq \ve,~\forall~x>L,~\forall~\xi>0,$$
	contradicted to \eqref{Vxk}.
	Notice that since such choice of $L$ is independent of $t$, we conclude that there exists $L\in[-\infty,+\infty)$ such that $[0,T]\times(L,+\infty)\subseteq\C$.	
\end{proof}

\begin{Lemma}
	For any $(t_0,x_0)\in\A$, the set $\Theta(t_0,x_0)$ is nonempty and $(t_0,x_0+\xi_0)\in\C$ for any $\xi_0\in\Theta(t_0,x_0)$.
\end{Lemma}

\begin{proof}
Since $V$ is uniformly bounded, one has
$$\lim_{\xi\rightarrow +\infty} V(t_0,x+\xi)-B(\xi) = -\infty, \lim_{\xi\rightarrow 0^+} V(t_0,x+\xi)-B(\xi) = V(t_0,x_0)-\kappa.$$
Then the condition $V(t_0,x_0) = \I V(t_0,x_0)$ implies that the supremum in $\I V(t_0,x_0)$ is achieved in the interior and thus $\Theta(t_0,x_0)$ is nonempty. 

	By property \eqref{subadd}, for any $\xi_0\in\Theta(t_0,x_0)$ one has
$$
\begin{array}{rl}
	\I V(t_0,x_0) &= \sup_{\xi\in\R^+}\left\{V(t_0,x_0+\xi)-B(\xi)\right\}\\[3mm]
	&\geq\sup_{\xi\in\R^+}\left\{V(t_0,x_0+\xi_0+\xi)-B(\xi_0+\xi)\right\}\\[3mm]
	&= \sup_{\xi\in\R^+}\left\{V(t_0,x_0+\xi_0+\xi)-B(\xi)\right\}-B(\xi_0)+\kappa\\[3mm]
	&=\I V(t_0,x_0+\xi_0)- B(\xi_0)+\kappa.
\end{array}
$$
On the other hand, since $\I V(t_0,x_0)+B(\xi_0) = V(t_0,x_0+\xi_0)$, we have
$$V(t_0,x_0+\xi_0)\geq \I V(t_0,x_0+\xi_0)+\kappa,$$
which implies that $x+\xi_0\in\C$.
\end{proof}

\begin{Lemma}
	Fix any $(t_0,x_0)\in\A$ and for any $\xi_0\in \Theta(t_0,x_0)$
	, one has
	$$V_x(t_0,x_0) = V_x(t_0,x_0+\xi_0) = 1.$$
\end{Lemma}
\begin{proof}
	By the definition of $\Theta(t,x)$, $\xi_0$ is a global maximum of the function $\xi\mapsto V(t_0,x_0+\xi)-B(\xi)$. Thus the first order condition yields that
	$$V_x(t_0,x_0+\xi_0) = B'(\xi_0) = 1.$$
	On the other hand, for any $\delta\neq 0$, we have
	$$V(t_0,x_0+\delta)\geq \I V(t_0,x_0+\delta)\geq V(t_0,x_0+\delta+\xi_0)-B(\xi_0).$$
	So one has
$$
\begin{array}{l}
\ds\frac{V(t_0,x_0+\delta)-V(t_0,x_0)}{\delta}\geq \frac{V(t_0,x_0+\delta+\xi_0)-V(t_0,x_0+\xi_0)}{\delta},~\delta>0\\[3mm]
\ds\frac{V(t_0,x_0+\delta)-V(t_0,x_0)}{\delta}\leq \frac{V(t_0,x_0+\delta+\xi_0)-V(t_0,x_0+\xi_0)}{\delta}.~\delta<0.	
\end{array}	
$$
By \eqref{Thm:42}, $V_x(t,x)$ is well defined for all $(t,x)\in(0,T)\times\R$. Taking $\delta\rightarrow 0^+$ and $\delta\rightarrow 0^-$, one achieves that
$$V_x(t_0,x_0) = V_x(t_0,x_0+\xi_0) = 1.$$
\end{proof}

\cleardoublepage

\end{document}